\documentclass[11pt,letterpaper]{article}
\usepackage[margin=1in]{geometry}
\usepackage{fullpage}
\usepackage{enumitem}
\usepackage{mathrsfs}
\usepackage[T1]{fontenc}
\usepackage{lmodern}
\usepackage{verbatim}
\usepackage{bm}
\usepackage{tikz-cd}
\usepackage{algorithm}
\usepackage{xcolor}
\usepackage{stmaryrd}
\usepackage{makecell}
\usepackage{dsfont}
\usepackage[normalem]{ulem}
\usepackage[hypertexnames=false]{hyperref}

\newcommand{\suppress}[1]{}
\usepackage{amsmath,amssymb,amsthm}
\usepackage{thmtools}

\makeatletter
\@ifundefined{newcounteralias}{}{%
  \renewcommand\thmt@autorefsetup{%
    \@xa\def
    \csname\thmt@envname autorefname\@xa\endcsname
    \@xa{\thmt@thmname}%
  }%
}
\makeatother

\newtheorem{theorem}{Theorem}[section]
\newtheorem{corollary}[theorem]{Corollary}
\newtheorem{lemma}[theorem]{Lemma}
\newtheorem{claim}[theorem]{Claim}
\newtheorem{definition}[theorem]{Definition}
\newtheorem{proposition}[theorem]{Proposition}

\DeclareMathOperator{\im}{Im}

\DeclareMathOperator{\rank}{rank}
\DeclareMathOperator{\Gr}{Gr}
\DeclareMathOperator{\Alt}{Alt}
\DeclareMathOperator{\sgn}{sgn}

\title{Polynomial-time Tensor Decomposition under the Kruskal Condition and the Lovitz--Petrov Generalization}
\author{
Shiri Sivan\\
Department of Computing and Mathematical Sciences, Caltech, USA\\
\texttt{ssivan@caltech.edu}\\
\href{https://orcid.org/0009-0004-3707-4844}{ORCID: 0009-0004-3707-4844}
\thanks{Supported in part by NSF CCF-2321079.}
}

\begin{document}
\pagenumbering{gobble}
\maketitle
\thispagestyle{empty}

\abstract{An identifiability criterion is a condition on a tensor decomposition that certifies uniqueness when the condition holds. Kruskal's classical condition is one of the best-known deterministic criteria for identifiability. However, no polynomial-time decomposition algorithm is known under the Kruskal condition; verifying the condition itself is NP-hard. Lovitz and Petrov introduced a strictly more general identifiability condition which, in contrast, is polynomial-time verifiable, but no polynomial-time decomposition algorithm was known under this condition either.

In this paper we give a randomized polynomial-time algorithm, in the algebraic model over $\mathbb R$, for tensor decomposition under the Lovitz--Petrov condition, and therefore also under the stricter Kruskal condition. 

Moreover, combining our algorithm with polynomial-time verification of the Lovitz--Petrov condition yields an efficient end-to-end certification procedure: after computing a decomposition, one can deterministically certify in polynomial time that it is unique and therefore of minimum rank.}

\clearpage

\tableofcontents
\thispagestyle{empty}
\clearpage

\pagenumbering{arabic}
\setcounter{page}{1}

\section{Introduction}

Every tensor admits a representation as a sum of rank-one tensors,\footnote{A nonzero tensor is rank one if it can be written as an outer product of vectors.} called a \emph{tensor decomposition}. The rank of a tensor is the minimum number of rank-one terms in such a decomposition. Tensor decompositions provide a natural way of extracting low-dimensional structure from multidimensional data and have found applications in areas including signal processing, machine learning, psychometrics, chemometrics, computer vision, and data analysis \cite{KoldaBader2009}.

Despite generalizing matrix factorization (2-way tensor decomposition), tensor decompositions of order \(m\ge 3\) behave quite differently from their matrix counterparts. While the rank of a matrix can be computed efficiently by Gaussian elimination, computing the rank of a tensor is NP-hard, already over \(\mathbb Q\) \cite{Hastad1990}, and remains NP-hard over \(\mathbb R\) and \(\mathbb C\) \cite{HillarLim2013}. On the other hand, tensors possess a useful property that matrices lack: \emph{uniqueness of decomposition}. A rank-one matrix has unique decomposition up to scaling, whereas higher-rank matrix decompositions are non-unique. In contrast, tensor decompositions can remain unique, up to permutation and scaling, at substantially higher ranks.

This uniqueness has concrete applications. In latent-variable models, uniqueness of tensor decomposition can imply identifiability of the statistical model \cite{AllmanMatiasRhodes2009}, while tensor decomposition algorithms can be used to recover the corresponding latent parameters \cite{AnandkumarEtAl2014}. Similarly, in chemometrics, uniqueness of the PARAFAC decomposition removes the rotational ambiguity inherent in matrix factorizations and can make it possible to recover pure spectra and concentrations \cite{Bro1997}. Related identifiability considerations arise in blind source separation and other signal-processing applications \cite{SidiropoulosEtAl2017}.

A tensor decomposition algorithm outputs an $r$-term decomposition, which certifies that the tensor rank is at most $r$. An identifiability criterion, in contrast, certifies uniqueness of a decomposition when its condition holds (though uniqueness may hold even when the criterion does not).\footnote{Uniqueness implies rank: Uniqueness among $r$-term decompositions implies that the rank is $r$.} For such criteria to be algorithmically useful, one needs both decomposition and verification to be efficient. Kruskal's classical condition \cite{kruskal1977three} is the best-known identifiability criterion. However, verifying that it holds for a tensor is NP-hard (computing the Kruskal rank of a matrix is NP-hard~\cite{TillmannPfetsch2014} and reduces to verifying Kruskal's condition). Domanov and De Lathauwer gave an algebraic decomposition algorithm under Kruskal's condition \cite{domanov2014canonical}, but the algorithm is not polynomial-time unless a certain parameter is bounded.\footnote{Both known algorithms by \cite{domanov2014canonical} and \cite{bhargava2026algorithmic} run in exponential time when the components are "balanced" (all have approximately the same Kruskal rank).}

The problem of giving an efficient algorithm for tensor decomposition under the Kruskal condition has remained unresolved since \cite{kruskal1977three}. In their discussion of this open problem, \cite{bhaskara2014open} offered a \$100 prize for answering the question:
\begin{center}
\emph{Is there a polynomial-time decomposition algorithm under the Kruskal condition?}
\end{center}
In this paper we provide an affirmative answer to this question.

The resolution of this question begins with a relaxation of Kruskal's identifiability criterion by Lovitz and Petrov~\cite{LovitzPetrov2023}, which we refer to as Condition \eqref{eq: intro m-LP}. In contrast to Kruskal's condition, Condition~\eqref{eq: intro m-LP} is polynomial-time verifiable, as observed by Claim~\ref{claim: LP verification}, but no decomposition algorithm was previously known under it. Our main result actually gives a polynomial-time algorithm under the weaker \eqref{eq: intro m-LP} condition.

Finally, combining our decomposition algorithm with polynomial-time verification of Condition \eqref{eq: intro m-LP} yields an efficient end-to-end certification procedure: after computing a decomposition, one can certify in polynomial time that it is the unique rank decomposition.

\subsection{Main results}

An \(r\)-decomposition of an \(m\)-way tensor
\(\mathcal T\in\mathcal A^{(1)}\otimes\cdots\otimes\mathcal A^{(m)}\)
is an expression of \(\mathcal T\) as a sum of \(r\) nonzero rank-one tensors,
\[
\mathcal T
=
\sum_{i=1}^r
a_i^{(1)}\otimes\cdots\otimes a_i^{(m)}.
\]

We refer to \(\mathcal A^{(j)}\) as the \(j\)-th \emph{mode} of \(\mathcal T\). Alternatively, we write
\[
\mathcal T
=
\llbracket A^{(1)},\ldots,A^{(m)}\rrbracket,
\]
where
\[
A^{(j)}
=
\begin{bmatrix}
a_1^{(j)}|\cdots|a_r^{(j)}
\end{bmatrix}
\in\left(\mathcal{A}^{(j)}\right)^{\times r}.
\]

The smallest \(r\) for which an $r$-decomposition exists is the \emph{tensor rank} of $\mathcal T$. A decomposition attaining this minimum is called a \emph{rank decomposition}.
Throughout, we assume \(m\ge 3\).

For \(S\subseteq[r]\), let \(A^{(j)}[S]\) denote the restriction of \(A^{(j)}\) to the columns indexed by \(S\).\\

The Lovitz--Petrov Condition is,
\begin{equation}
\label{eq: intro m-LP}
\sum_{j=1}^m \rank\left(A^{(j)}[S]\right)
\geq
2|S|+m-1
\qquad
\text{for every }S\subseteq[r],\ |S|\geq2.
\tag{LP}
\end{equation}

Lovitz and Petrov \cite{LovitzPetrov2023} prove that Condition~\eqref{eq: intro m-LP} implies uniqueness of the rank-\(r\)
decomposition. In contrast to the Kruskal condition, the \eqref{eq: intro m-LP} Condition allows substantial rank deficiencies in the individual factor matrices and does not require any prescribed pair of modes to have sufficiently large Kruskal rank.

\paragraph{Computational model.}
We work in the algebraic model over $\mathbb{R}$ and measure the number
of arithmetic operations. We denote by $N$ the input size, namely the number of scalar entries of $\mathcal T$. In addition to field operations, we allow exact root finding for univariate polynomials, with running time polynomial in the degree. Random covectors are sampled from absolutely continuous distributions over the relevant real vector spaces. We say that a property holds \emph{generically} on a parameter space if it holds outside a proper algebraic subset of that space. Thus, under the sampling scheme above, generic properties hold with probability one.

\begin{theorem}[Main theorem]
\label{thm: main}
Let
\[
\mathcal T
=
\llbracket A^{(1)},\ldots,A^{(m)}\rrbracket
\in
\mathbb R^{n_1\times\cdots\times n_m}
\]
be an \(m\)-way tensor with \(r\) rank-one summands satisfying the \eqref{eq: intro m-LP} Condition. By the Lovitz--Petrov theorem \cite[Theorem~2]{LovitzPetrov2023}, this decomposition is unique up to permutation and scaling of the columns.\footnote{We give a separate proof of uniqueness, based on the spectral structure developed in this paper, in
Section~\ref{sec: uniqueness}.}

Moreover, there is a randomized algorithm that, with probability one, recovers the factor matrices up to these ambiguities and runs in
\(\operatorname{poly}(N)\) time under the algebraic model over $\mathbb{R}$, where $N:=\prod_{j=1}^m n_j$ is the number of entries of the tensor.
\end{theorem}

A reference implementation of the algorithm is available on GitHub.\footnote{\url{https://github.com/shirisiv/LP-tensor-decomposition}} 
Since Kruskal's condition implies \eqref{eq: intro m-LP} (see Proposition \ref{prop: kruskal implies LP} for definition and proof), we obtain a resolution of the aforementioned open question.

\begin{corollary}[Kruskal decomposition]
\label{cor: kruskal decomposition}
There is a randomized polynomial-time algorithm, in the algebraic model over \(\mathbb R\), that recovers any tensor decomposition satisfying Kruskal's condition, up to permutation and componentwise scaling, with probability one.
\end{corollary}

Theorem~\ref{thm: main} gives polynomial-time recovery under the promise that the unknown decomposition satisfies Condition~\eqref{eq: intro m-LP}. Since this promise can itself be verified in polynomial time after recovery (by Claim \ref{claim: LP verification}), we obtain the following certification result.

\begin{corollary}[Polynomial-time certification]
\label{cor: certification}
Given a decomposition
\[
\mathcal T=\llbracket A^{(1)},\ldots,A^{(m)}\rrbracket\in\mathbb{R}^{n_1\times\dots\times n_m},
\]
Condition~\eqref{eq: intro m-LP} can be verified deterministically in polynomial time. Consequently, after recovering a decomposition as in Theorem~\ref{thm: main}, one can certify in polynomial time that it is the unique rank decomposition of $\mathcal T$.
\end{corollary}

\paragraph{Relation to the Jennrich method.}
The central step of our algorithm is an eigendecomposition that separates the individual terms of the tensor decomposition. Each term is associated with a two-dimensional eigenspace, and the construction yields distinct eigenvalues on these eigenspaces, thereby breaking the symmetry between the terms. The role of distinct eigenvalues here is reminiscent of the symmetry-breaking mechanism underlying the celebrated Jennrich algorithm~\cite{leurgans1993decomposition}. This analogy is conceptually useful, but concerns only the spectral step: the geometric ideas needed to produce the required spectral structure under the Lovitz--Petrov condition are substantially different from those underlying Jennrich's method. Section~\ref{sec: geometric perspective} develops both the connection and the differences in detail.

\paragraph{Organization.}
Section~\ref{sec: preliminaries} introduces the notation and alternating representation used throughout the paper. Section~\ref{sec: decomposition} develops the alternating-contraction framework, presents the decomposition algorithm, and proves Theorem~\ref{thm: main}. Section~\ref{sec: verification} proves that Condition~\eqref{eq: intro m-LP} is polynomial-time verifiable, and Section~\ref{sec: uniqueness} gives an algorithmic uniqueness proof. Section~\ref{sec: geometric perspective} gives a geometric overview of the algorithm and develops its analogy with the celebrated Jennrich algorithm. Auxiliary results and deferred proofs appear in the \hyperref[sec: appendix]{Appendix}.

\section{Preliminaries}
\label{sec: preliminaries}

\paragraph{Notation.} Generally, a space is denoted by a calligraphic letter $\mathcal{M}$, and a corresponding matrix is denoted by a capital letter $M$. For a matrix $E$, we denote by $E[S]$ the submatrix consisting of the columns indexed by $S$, and by $E[\widehat S]$ the submatrix obtained by removing the columns indexed by $S$.

\paragraph{Problem setup.}
Throughout the paper, $\mathcal T=\sum_{i=1}^r \mathcal T_i\in\mathbb R^{n_1\times\cdots\times n_m}$ denotes an \(m\)-way tensor with an \(r\)-decomposition, where \(m\ge 3\) and
\[
\mathcal T_i=
a_i^{(1)}\otimes\cdots\otimes a_i^{(m)},
\qquad
a_i^{(j)}\in\mathbb R^{n_j}\setminus\{0\}.
\]
For \(j\in[m]\), set
\[
\mathcal A^{(j)}:=\mathbb R^{n_j}
\]
and denote
\[
A^{(j)}
=
\begin{bmatrix}
a_1^{(j)}\vert\cdots\vert a_r^{(j)}
\end{bmatrix}
\in\mathbb{R}^{n_j\times r}.
\]

We fix
\begin{equation}
\label{notation: X,D,d}
X=\bigoplus_{j=1}^m\mathcal A^{(j)},
\qquad
D=\dim X=\sum_{j=1}^m n_j,
\qquad
d=m-2.
\end{equation}

For every \(i\in[r]\), define
\begin{equation}
\label{eq: F_i def}
F_i
:=
\begin{bmatrix}
\bar a_i^{(1)}\vert\cdots\vert\bar a_i^{(m)}
\end{bmatrix}
\in X^{\times m},
\end{equation}
where \(\bar a_i^{(j)}\in X\) denotes the canonical embedding of
\(a_i^{(j)}\in\mathcal A^{(j)}\) into \(X\). Complying with our notational convention, we denote
\[
\mathcal F_i:=\im F_i,
\qquad i\in[r].
\]

\paragraph{The Lovitz--Petrov Condition.}
 We say that $\mathcal T$ satisfies Condition \eqref{con: m-LP} if
\begin{equation}
\sum_{j\in[m]}\rank A^{(j)}[S]\geq 2|S|+m-1
\qquad
\text{for every }S\subseteq[r]\text{ with }|S|\geq2.
\tag{LP}
\label{con: m-LP}
\end{equation}

We provide a useful equivalent formulation. Since the columns of $F_i$ lie in distinct direct summands of $X$, we have
\[
R(S):=\dim\left(\sum_{i\in S}\mathcal{F}_i\right)=\sum_{j=1}^m\rank A^{(j)}[S]
\]
for all $S\subseteq[r]$.\footnote{The notation $R(S)$ hints at the fact that $([r],R)$ forms a (linear) polymatroid over the collection 
$\{\mathcal{F}_i\subseteq X\}_{i\in [r]}$ of linear spaces.} We obtain an equivalent formulation of the Lovitz-Petrov condition.
\begin{equation}
\label{eq: m-LP F-form}
R(S)\geq 2|S|+m-1
\qquad
\text{for every }S\subseteq[r]\text{ with }|S|\geq2.
\tag{R}
\end{equation}

\subsection{Alternating representations and contractions}
\label{sec: alternating rep}

\paragraph{Alternating tensors.}
For a vector space $X$ and $k\geq 1$, let $\Pi_\sigma:X^{\otimes k}\to X^{\otimes k}$ be the linear map defined by
\[
\Pi_\sigma(x_1\otimes\cdots\otimes x_k)
:=
x_{\sigma(1)}\otimes\cdots\otimes x_{\sigma(k)}.
\]
The space of \emph{alternating $k$-tensors} on $X$ is
\[
\bigwedge^k X
:=
\left\{
T\in X^{\otimes k}:
\Pi_\sigma T
=
\operatorname{sgn}(\sigma)\,T
\text{ for every }\sigma\in S_k
\right\}.
\]

Define\footnote{Some authors define the normalized sum $x_1\wedge\cdots\wedge x_k
:=
\frac{1}{k!}\sum_{\sigma\in S_k}
\operatorname{sgn}(\sigma)\,
x_{\sigma(1)}\otimes\cdots\otimes x_{\sigma(k)}$.} $\Alt(x_1\otimes\cdots\otimes x_k)
:=
x_1\wedge\cdots\wedge x_k
:=
\sum_{\sigma\in S_k}
\operatorname{sgn}(\sigma)\,
x_{\sigma(1)}\otimes\cdots\otimes x_{\sigma(k)}$,
and extend $\Alt:X^{\otimes k}\to\bigwedge^k X$ linearly.

For $E=[e_1,\dots,e_k]\in X^{\times k}$, we write\footnote{While $\Alt$ is linear in $X^{\otimes{k}}$, it is multilinear and alternating in $X^{\times k}$.}
\[
\Alt(E):=\Alt(e_1\otimes\cdots\otimes e_k).
\]
An element of \(\bigwedge^kX\) of this form is called \textit{decomposable}.

Using the canonical embeddings \(\mathcal A^{(j)}\hookrightarrow X\), let
\[
\bar{\mathcal T}_i
:=\bar a_i^{(1)}\otimes\cdots\otimes\bar a_i^{(m)}
\in X^{\otimes m},
\qquad
\bar{\mathcal T}:=\sum_{i=1}^r\bar{\mathcal T}_i.
\]
Thus
\[
\Alt(\bar{\mathcal T})
=
\sum_{i=1}^r\Alt(\bar{\mathcal T}_i)
=
\sum_{i=1}^r\Alt(F_i),
\]
where $F_i$ is defined in Equation \eqref{eq: F_i def}.

\paragraph{Contractions.}
For $p\in X^*$, define the contraction $\iota_p:\bigwedge^k X\longrightarrow \bigwedge^{k-1}X$ on decomposable elements by
\[
    \iota_p(\Alt(E))
    =
    \sum_{j=1}^k
    (-1)^{j-1}p(e_j)\Alt\left( E[\;\widehat{j}\;]\right).
\]
This extends uniquely to a linear map on $\bigwedge^kX$. More generally, for
$P=
    \begin{bmatrix}
        p_1\\
        \vdots\\
        p_t
    \end{bmatrix}
    \in (X^*)^{\times t}$,
define
$\iota_P
    :=\iota_{p_t}\circ\cdots\circ\iota_{p_1}:
    \bigwedge^k X\longrightarrow\bigwedge^{k-t}X$. 

\begin{lemma}[Basic properties of alternating representations]
\label{lem: alternating properties}
Let $E=[e_1,\dots, e_m]\in X^{\times m}$, and let $d:=m-2$. Then the following hold.

\begin{enumerate}
    \item \emph{Change of basis.}
    For every $N\in\mathbb{R}^{m\times m}$,
    \[
        \operatorname{Alt}(EN)
        =
        \det(N)\operatorname{Alt}(E).
    \]
    In particular, if $N\in\operatorname{GL}_m$, then   $\langle\operatorname{Alt}(EN)\rangle=\langle\operatorname{Alt}(E)\rangle$.
    \item \emph{Coordinate formula for contractions.}
    Let
    \[
        P=
        \begin{bmatrix}
            p_1\\
            \vdots\\
            p_d
        \end{bmatrix}
        \in (X^*)^{\times d},
    \]
    then
    \begin{equation}
    \label{eq: contraction formula}
        \iota_P\operatorname{Alt}(E)
        =
        \sum_{1\leq j<k\leq m}
        (-1)^{j+k+1}
        \det\!\left(PE[\widehat{j,k}]\right)\Alt(E[j,k]).
    \end{equation}
\end{enumerate}
\end{lemma}

We defer the proof of Lemma \ref{lem: alternating properties} to the \hyperref[proof: alternating properties]{Appendix}.

\section{Decomposition under the Lovitz-Petrov Condition}
\label{sec: decomposition}

The central tool of the algorithm is alternating contractions of the tensor. These are used to transform the tensor decomposition into a spectral decomposition problem. The algorithm can be viewed as an alternating and relativized variant of the classical Jennrich algorithm~\cite{leurgans1993decomposition}; see Section~\ref{sec: geometric perspective} for a review of Jennrich's algorithm and its relation to this work. In Section~\ref{subsec: alternating contractions}, we introduce notation for alternating contractions, show how to compute them efficiently from the input tensor, and explain how a rank-one component can be recovered from its contraction. In Section~\ref{sec: decomposition alg}, we present the decomposition algorithm and prove Theorem~\ref{thm: main}.

\subsection{Alternating contractions}
\label{subsec: alternating contractions}

\begin{definition}[Alternating contraction]
\label{def: alternating contraction}
For every $P=
\begin{bmatrix}
p_1\\
\vdots\\
p_d
\end{bmatrix}
\in\mathbb R^{d\times D}$, define
\begin{equation}
    \label{eq: Omega_P definition}
    \Omega_P(\mathcal{T}):=\iota_P(\Alt(\bar{\mathcal T)})\in\bigwedge^2 X
\end{equation}
In particular, under the standard identification $x\otimes y\longleftrightarrow xy^\top$,
\[\Omega_P(\mathcal{T})\in \mathbb R^{D\times D}.\]
Since the contraction is a linear operator,
$\Omega_P(\mathcal{T})=\sum_{i=1}^r \Omega_P(\mathcal{T}_i)$.
\end{definition}

While $\Omega_P$ is linear in the tensor space, as a function of $P$ (with fixed $\mathcal{T}$) it is multilinear and alternating in the rows of $P$.
The following Corollary of Lemma \ref{lem: alternating properties} provides a practical formula for computing contractions.

\begin{corollary}[Computational form]
\label{cor: computational form}
For every $j\in[m]$, let $\{e_1^{(j)},\ldots,e_{n_j}^{(j)}\}$ denote the standard basis of $\mathcal A^{(j)}\cong\mathbb{R}^{n_j}$, and write
\[
\mathcal T
=
\sum_{\alpha\in[n_1]\times\cdots\times[n_m]}
\mathcal T[\alpha]\,
e_{\alpha_1}^{(1)}
\otimes\cdots\otimes
e_{\alpha_m}^{(m)}.
\]
For every multi-index $\alpha=(\alpha_1,\ldots,\alpha_m)\in [n_1]\times\dots\times[n_m]$, define
\[
E_\alpha
:=
\begin{bmatrix}
\bar e_{\alpha_1}^{(1)}
\mid\cdots\mid
\bar e_{\alpha_m}^{(m)}
\end{bmatrix}
\in\mathbb R^{D\times m}
\]
Then, 
\begin{equation}
\label{eq: computational omega}
\Omega_P(\mathcal{T})=
\sum_{\alpha}
\mathcal T[\alpha]\sum_{1\leq j<k\leq m}\,
(-1)^{j+k+1}\det\left(PE_\alpha[\widehat{j,k}]\right)\Alt\left(E_\alpha[{j,k}]\right).
\end{equation}
Hence $\Omega_P(\mathcal{T})$ can be computed directly from the input tensor and $P$,
without knowing the rank decomposition. Furthermore, there are only $N\binom{m}{2}=\mathrm{poly}(N)$ summands in this expression.
\end{corollary}

Note that $\Omega_P(\mathcal{T})$ is directly computable via Corollary \ref{cor: computational form}, while the component contractions $\Omega_P(\mathcal{T}_i)$ are not. In fact, computing these terms is the main objective of the decomposition algorithm, as the column vectors $a_i^{(j)}$ can be directly recovered from them (up to scaling), as per the following corollary (of Lemma \ref{lem: alternating properties}). 

\begin{corollary}
\label{cor: block structure of omega(T_i)}
Let $\mathcal T_i$ be as in the problem setup. Then
$\Omega_P(\mathcal T_i)\in\mathbb R^{D\times D}$ is an $m\times m$ skew-symmetric block matrix whose diagonal blocks vanish and 
\[
\Omega_P(\mathcal T_i)[j,k]
=
(-1)^{j+k+1}
\det\!\bigl(PF_i[\widehat{j,k}]\bigr)\,
a_i^{(j)}\bigl(a_i^{(k)}\bigr)^\top
\in
\bigl\langle a_i^{(j)}\otimes a_i^{(k)}\bigr\rangle
\]
for all $1\le j<k\le m$. The blocks below the diagonal are determined by skew-symmetry.

Therefore, whenever $\det\!\bigl(PF_i[\widehat{j,k}]\bigr)\neq 0$, the vectors $a_i^{(j)}$ and $a_i^{(k)}$ can be recovered from
$\Omega_P(\mathcal T_i)[j,k]$ up to scaling.
\end{corollary}

\subsection{A Decomposition algorithm under the LP Condition}
\label{sec: decomposition alg}

If \(A\in\operatorname{GL}_d\), then  $\Omega_{AP}(\mathcal T)= \det(A)\,\Omega_P(\mathcal T)$. Thus the image and kernel of \(\Omega_P(\mathcal T)\) depend only on the
row space of \(P\), and we will freely regard \(P\) as a point of
\(\Gr(d,X^*)\) whenever only these spaces are relevant.

\begin{algorithm}[H]
    \caption{Rank decomposition under Condition LP}
    \label{alg: LP-decomposition}

    \textbf{Input:} A tensor $\mathcal{T}$ with unknown $r$-decomposition
    \[
        \mathcal{T}=\llbracket A^{(1)},\dots,A^{(m)}\rrbracket\in\mathcal{A}^{(1)}\otimes\dots\otimes\mathcal{A}^{(m)}
    \]
    satisfying Condition \eqref{con: m-LP}. Let $X=\mathcal{A}^{(1)}\oplus\dots\oplus\mathcal{A}^{(m)}$, $D:=\dim X$, and $d:=m-2$.

    \begin{enumerate}
        \item\label{mLP step: comp Omega_P} Choose random $P\in \operatorname{Gr}(d,X^*)$, and construct $\Omega_P:=\Omega_P(\mathcal{T})\in\mathbb{R}^{D\times D}$ (per Equation \eqref{eq: computational omega}). 
        \item\label{mLP step: comp Omega_Q} Choose random $Q\in\operatorname{Gr}(d, \ker\Omega_P)$, and construct  $\Omega_Q:=\Omega_Q(\mathcal{T})\in\mathbb{R}^{D\times D}$ (per Equation \eqref{eq: computational omega}).

        \item\label{mLP step: comp L} Compute $L\in\mathbb{R}^{D \times \ell}$, an orthonormal basis for $\im\Omega_P$, where $\ell:=\rank\Omega_P$.

        \item\label{mLP step: comp tildes and phi} Compute
        \[
            \widetilde{\Omega}_P:=L^\top\Omega_PL,
            \qquad
            \widetilde{\Omega}_Q:=L^\top\Omega_QL,
        \]
        and set $\Phi:=
            \widetilde{\Omega}_Q
            \widetilde{\Omega}_P^{-1}$.

        \item\label{mLP step: comp eigen-projection} Compute the distinct eigenvalues $ \{\rho_1,\ldots,\rho_{\widehat{r}}\}$ of $\Phi$. For every $i\in[\widehat{r}]$, compute the spectral projector
        \[
            \Pi_i:=\prod_{j\neq i}\frac{\Phi-\rho_j I}{\rho_i-\rho_j}\in\mathbb{R}^{\ell\times\ell}.
        \]

        \item\label{mLP step: recover omega_i} For every $i\in[\widehat{r}]$, recover the corresponding component matrix
        \[
           \widehat{\Omega}_i=
            L\Pi_i\widetilde{\Omega}_PL^\top\in\mathbb{R}^{D\times D}.
        \]
        \item\label{mLP step: recover ai's}
        For every \(i\in[\widehat r]\), view \(\widehat\Omega_i\) as an \(m\times m\) block matrix according to $X=\bigoplus_{j=1}^m \mathcal A^{(j)}$. If the blocks \(\widehat\Omega_i[1,j]\), \(j=2,\ldots,m\), are nonzero rank-one matrices, factor them to obtain
        \[
        \widehat a_i^{(j)}\in\mathcal A^{(j)},\qquad\forall j\in[m]
        \]
        up to scaling, choosing a single representative \(\widehat a_i^{(1)}\) from any one block.
        
        \item\label{mLP step: recover lambdas} Solve  $ \mathcal{T}=\sum_{i\in[\widehat{r}]}\lambda_i \widehat a_i^{(1)}\otimes\dots\otimes \widehat a_i^{(m)}$ for $\lambda\in\mathbb{R}^{\widehat{r}}$.
        \item Update $\widehat a_i^{(1)}\leftarrow \lambda_i\widehat{a}_i^{(1)}$ for all $i\in[\widehat{r}]$.
    \end{enumerate}
    \textbf{Output:} Decomposition factors $\widehat{A}^{(1)},\dots,\widehat{A}^{(m)}$, equivalent to $A^{(1)},\dots,A^{(m)}$ up to permutation and scaling of the columns.
\end{algorithm}

An implementation of Algorithm~\ref{alg: LP-decomposition} is available at \\
\url{https://github.com/shirisiv/LP-tensor-decomposition}.

\begin{theorem}[Algorithmic LP]
\label{thm: algorithmic LP}
Let
\[
\mathcal{T}
=
\llbracket A^{(1)},\dots,A^{(m)}\rrbracket
\in
\mathbb{R}^{n_1\times\cdots\times n_m}
\]
be an $r$-term decomposition satisfying Condition~\eqref{con: m-LP}.
Then by the Lovitz--Petrov theorem \cite[Theorem~2]{LovitzPetrov2023}, this decomposition is a unique rank decomposition.

Moreover, for generic choices of the contractions in
Algorithm~\ref{alg: LP-decomposition}, the algorithm recovers $A^{(1)},\dots,A^{(m)}$ up to permutation and componentwise scaling.
Equivalently, if the contractions are sampled from absolutely continuous distributions, the algorithm succeeds with probability one.

Let $N=\prod_{j=1}^m n_j$ be number of entries of $\mathcal{T}$, then Algorithm \ref{alg: LP-decomposition} runs in $\operatorname{poly}(N)$ time under the algebraic model over $\mathbb{R}$.
\end{theorem}

\begin{proof}
    \textbf{Correctness:} Uniqueness is established by \cite[Theorem~2]{LovitzPetrov2023} (or alternatively, by Lemma \ref{lem: uniqueness}). Thus, we only need to show that Algorithm \ref{alg: LP-decomposition} outputs the components $A^{(j)}$ (up to scaling and permutation) with probability 1. Let $\mathcal{L}_i:=\ker P\cap\mathcal{F}_i$ for all $i\in[r]$, and $\mathcal{L}:=\sum_{i=1}^r\mathcal{L}_i$. For every $i$, set $\mathcal{R}_i$ a complement $\mathcal{F}_i=\mathcal{L}_i\oplus\mathcal{R}_i$, and $L_i$, $R_i$ arbitrary bases (respectively). Before proving correctness we establish a few facts. 
    \begin{enumerate}
        \item $P\in\Gr(d,\mathcal{L}^\perp)$ because $\mathcal{L}\subset\ker P$. By choice of $R_i$, $\det(PR_i)\neq 0$ for all $i$. Applying Corollary~\ref{cor: L-adapted} with $\mathcal{V}=\ker P$, yields
        \[
        \im\Omega_P(\mathcal{T}_i)=\mathcal{L}_i,\qquad\dim\mathcal{L}_i=2\qquad\forall i\in[r],
        \]
        and
        \begin{equation}
        \label{eq: image direct sum}
        \im\Omega_P(\mathcal{T})=\mathcal{L}=\bigoplus_{i=1}^r\mathcal{L}_i.
        \end{equation}
        \item Since $\Omega_P$ is skew-symmetric, $\ker \Omega_P=\left(\im\Omega_P\right)^\perp=\mathcal{L}^\perp$, so $Q$ is chosen generically in $\Gr(d,\mathcal{L}^\perp)$. By Lemma \ref{lem: spectral-separation}, $\det(QR_i)\neq 0$ for all $i$, so applying Corollary~\ref{cor: L-adapted} to $Q$, leads to
         \[
        \im\Omega_Q(\mathcal{T}_i)=\mathcal{L}_i,\qquad        \im\Omega_Q(\mathcal{T})=\mathcal{L}.
        \]
        
        \item  By Lemma \ref{lem: spectral-separation}, 
        \[
        \rho'_i:=\frac{\det(QR_i)}{\det(PR_i)}\qquad i\in[r]
        \]
        are distinct.
        \item By Corollary \ref{cor: L-adapted}, for some $M_i\in\operatorname{GL}_m$,
        \[
        \Omega_P(\mathcal{T}_i)=\det(M_i)\det(PR_i)\Alt(L_i),\qquad
        \Omega_Q(\mathcal{T}_i)=\det(M_i)\det(QR_i)\Alt(L_i)
        \]
        so 
        \begin{equation}
        \label{eq: rho dec}
        \Omega_Q(\mathcal{T})=\sum_{i=1}^r\rho'_i\Omega_P(\mathcal{T}_i).
        \end{equation}
    \end{enumerate}
    We are now ready to prove correctness.
    \begin{itemize}
        \item \textbf{$\Phi$ is well defined.} The column and row space of a skew-symmetric matrix are equal, so $\widetilde{\Omega}_P\in\mathbb{R}^{2r\times 2r}$ is invertible.
        \item \textbf{$\Phi$ has $r$ distinct, non-zero eigenvalues $\rho_i=\rho'_i$.}
        Denote for all $i\in[r]$,
         \[\mathcal E_i:=L^\top({\mathcal{L}_i}).\]
         By Equation \eqref{eq: image direct sum}, and since $L^\top\in\mathbb{R}^{2r\times D}$ is bijective on $\mathcal{L}$,
        \[
        \im\Phi=\bigoplus_i\mathcal{E}_i.
        \]
        Furthermore, $\rho'_i$ is an eigenvalue associated with $\mathcal{E}_i$:

        Let $z\in\mathcal{E}_i$, and $x:=\widetilde{\Omega}_P^{-1}z$.
        Then $Lz=L\widetilde{\Omega}_Px=\Omega_PLx\in\mathcal{L}_i$.
        By directness,
        \[
        \Omega_P(\mathcal T_j)Lx=0\quad \forall j\neq i.
        \]
        Applying Equation \eqref{eq: rho dec} we get
        \[
        \Phi z
        =
        \widetilde{\Omega}_Q\widetilde{\Omega}_P^{-1}z
        =
        L^\top\Omega_QLx
        =
        L^\top\sum_j\rho'_j\Omega_P(\mathcal{T}_j)Lx
        =
        \rho'_iL^\top Lz
        =
        \rho'_i z.
        \]
        Therefore, up to permutation, Step \ref{mLP step: comp eigen-projection} calculates  $\widehat{r}=r$ distinct eigenvalues $\rho_i=\rho'_i$ (with corresponding eigenspaces $\mathcal E_i$).
  
        \item  \textbf{Recovery of $\Omega_P(\mathcal{T}_i)$:} By linearity, $\Omega_P=\sum_{i=1}^r\Omega_P(\mathcal{T}_i)$. Therefore, for all $i\in[r]$,
    \[
    \widehat{\Omega}_i=
    L\Pi_i\widetilde{\Omega}_PL^\top=
    L\Pi_i(L^\top\Omega_P L)L^\top=
    L\Pi_iL^\top\sum_{j\in[r]}\Omega_P(\mathcal{T}_j)=\Omega_P(\mathcal{T}_i).
    \]

    \item \textbf{Recovery of the factor matrices.}
For every $i\in[r]$, $\rank F_i[\widehat{1,j}]=d$.
Hence, for generic $P\in\Gr(d,X^*)$,
simultaneously for every $i\in[r]$ and $j=2,\dots,m$,
\[
\det\!\left(PF_i[\widehat{1,j}]\right)\neq 0.
\]

Thus, by Corollary \ref{cor: block structure of omega(T_i)}, $a_i^{(1)}, a_i^{(j)}$ can be recovered from the block $\Omega_P(\mathcal{T}_i)[1,j]$ up to scaling. After absorbing $\lambda_i$ into $\widehat a_i^{(1)}$, the output components are equivalent to
$A^{(1)},\dots,A^{(m)}$ up to scaling and
permutation.
    \end{itemize}

\textbf{Time Complexity:}
We may assume without loss of generality that \(n_j\ge 2\) for every \(j\in[m]\), since a one-dimensional mode consists only of scalar factors and can be removed in preprocessing without affecting Condition~\eqref{con: m-LP}. Consequently,
\[
N=\prod_{j=1}^m n_j\ge 2^m,
\]
so $m\le \log_2 N$. Moreover,
\[
D=\sum_{j=1}^m n_j\le \prod_{j=1}^m n_j=N.
\]
For any tensor $r\leq N$, therefore any computation polynomial in $D,m,$ and $r$ is polynomial in $N$.

Steps~\ref{mLP step: comp L} through~\ref{mLP step: recover omega_i} require a polynomial number of arithmetic operations on matrices of dimensions at most \(D\times D\), together with the computation of the eigenvalues of a \(2r\times2r\) matrix. Equivalently, the latter reduces to finding the roots of a univariate polynomial of degree at most \(2r\). Once the eigenvalues are known, the corresponding eigenspaces and spectral projectors can be computed using \(\operatorname{poly}(D,r)\) arithmetic operations.

The rank-one blocks in Step~\ref{mLP step: recover ai's} can be factored by taking any nonzero row and column and normalizing by their common nonzero entry. This requires \(\operatorname{poly}(D,r,m)\) arithmetic operations. Finally, Step~\ref{mLP step: recover lambdas}, after vectorizing the tensor, amounts to solving a linear system with \(N\) equations and \(r\) unknowns, and therefore requires \(\operatorname{poly}(N,r)\) arithmetic operations.

It remains to bound the cost of constructing $\Omega_P$ and
$\Omega_Q$. Each requires computing $\binom{m}{2}$ blocks
$\Omega_P(\mathcal{T})[j,k]$. By Corollary~\ref{cor: computational form},
\[
\Omega_P(\mathcal{T})[j,k]
=
(-1)^{j+k+1}
\sum_{\alpha\in[n_1]\times\dots\times[n_m]}
\mathcal T[\alpha]\,
\det\!\left(PE_\alpha[\widehat{j,k}]\right)
e_{\alpha_j}^{(j)}\otimes e_{\alpha_k}^{(k)}.
\]
For fixed $j,k$, there are $N$ summands, and each determinant is of an
$(m-2)\times(m-2)$ matrix. Thus $\Omega_P(\mathcal{T})[j,k]$ can be computed in
$O\!\left(Nm^3\right)$ arithmetic operations. Computing all $\binom{m}{2}$ blocks therefore takes
\[
N\cdot\operatorname{poly}(m)
\]
arithmetic operations. Since $m\le\log_2N$, $D\le N$, and $r\le N$, all steps of
Algorithm~\ref{alg: LP-decomposition} require \(\operatorname{poly}(N)\) arithmetic operations, except for Step \ref{mLP step: comp eigen-projection} which finds the roots of a univariate polynomial of degree at most $2r=O(N)$.
\end{proof}

We now establish the geometry underlying the spectral separation.

\begin{lemma}[Generic section]
\label{lem: generic-section}
Suppose that the \eqref{con: m-LP} Condition holds.
For generic $\mathcal{V}\in\Gr(D-d,X)$, define
\[
    \mathcal{L}_i:=\mathcal V\cap\mathcal{F}_i,
    \qquad
    \mathcal{L}:=\sum_{i=1}^r\mathcal{L}_i.
\]
For every $i$, choose some complement $\mathcal{R}_i$ s.t.
$\mathcal{F}_i=\mathcal{L}_i\oplus\mathcal{R}_i$. Then,
\[
    \dim\mathcal{L}_i=2\quad\text{and}\quad\dim\mathcal{R}_i=d
    \qquad\text{for every }i\in[r],
\]
\[
    \mathcal{L}
    =
    \bigoplus_{i=1}^r\mathcal{L}_i,
\]
and for every $j\neq i\in[r]$ there exists $0\neq v_{ij}\in\mathcal{R}_j$ such that
\[
    \mathcal{R}_i
    \oplus
    \langle v_{ij}\rangle
    \oplus
    \mathcal{L}
\]
is a direct sum.
\end{lemma}

\begin{proof}
The first statement follows from generic intersection, since
$\dim\mathcal{F}_i=m$ and $\mathcal V$ has codimension $d=m-2$.
For $\mathcal{L}=\bigoplus_{i=1}^r\mathcal{L}_i$ apply the \hyperref[thm: lovasz-dilworth]{Lov\'asz's Theorem} on the family $\{\mathcal{F}_i\}_{i\in[r]}$. For
\[
\widetilde{R}(S):=R(S)-d=\dim\sum_{i\in S}\mathcal{F}_i-d
\]
Equation \eqref{eq: lovasz} (with $S=[r]$) becomes
\[
\dim\left(\sum_{i=1}^r\mathcal{L}_i\right)
=
\min_{\Pi\vdash [r]}
\sum_{B\in\Pi}\widetilde{R}(B).
\]
By Condition \eqref{eq: m-LP F-form} (alternative formulation of \eqref{con: m-LP}), and $\dim\mathcal{F}_i=m$, the minimal contribution of $\mathcal{F}_i$ is attained on the singleton $B=\{i\}$, thus, $\dim\sum_{i=1}^r\mathcal{L}_i=2r$. Since $\dim\mathcal{L}_i=2$ for all $i$, the sum is direct.

For $\mathcal{R}_i\oplus\langle v_{ij}\rangle\oplus\mathcal{L}$ apply the Lov\'asz formula to the family
\[
    \{\mathcal{F}_i+\mathcal{F}_j\}
    \cup
    \{\mathcal{F}_s\}_{s\neq i,j}.
\]
The singleton corresponding to $\mathcal F_i+\mathcal F_j$ contributes at least 5, while every other singleton contributes 2. By \eqref{eq: m-LP F-form}, a block containing $\mathcal F_i+\mathcal F_j$ and $t$ other elements contributes at least $2t+5$. Hence
\[
\dim\left(\mathcal V\cap(\mathcal{F}_i+\mathcal{F}_j)+\sum_{s\neq i,j}\mathcal{L}_s\right)\geq 2(r-2)+5=2r+1.
\]

 We have
\[
\mathcal V\cap(\mathcal{F}_i+\mathcal{F}_j)+\sum_{s\neq i,j}\mathcal{L}_s\subseteq (\mathcal{F}_i+\mathcal{F}_j+\mathcal{L})\cap \mathcal V.
\]
Since 
\[\dim(\mathcal V\cap\mathcal F_i)=2,\quad
\dim\mathcal V^\perp|_{\mathcal F_i}=d\]
and hence
\[
\dim\mathcal V^\perp|_{\mathcal F_i+\mathcal F_j+\mathcal L}=d,
\]
the rank-nullity theorem implies that
\[
\dim\left(\mathcal{F}_i+\mathcal{F}_j+\mathcal{L}\right)=
\dim\left[\mathcal V\cap (\mathcal{F}_i+\mathcal{F}_j+\mathcal{L})\right]+d.
\]
Therefore
\[
    \dim(\mathcal{F}_i+\mathcal{F}_j+\mathcal{L})
    \geq 2r+1+d.
\]
Since $\dim\mathcal{R}_i=d$ and $\dim\mathcal{L}=2r$,  there must be some $ 0\neq v_{ij}\in\mathcal{R}_j$, which gives the claimed direct sum.
\end{proof}

\begin{corollary}[Contraction in an $\mathcal{L}$-adapted basis]
\label{cor: L-adapted}
For generic $\mathcal V\in\operatorname{Gr}(D-d, X)$, define $\mathcal{L}_i,\mathcal{R}_i$ and $\mathcal{L}$ as in Lemma \ref{lem: generic-section}. 
Then for every $P\in\Gr(d,\mathcal{L}^{\perp})$, and fixed bases $L_i, R_i$, there exists a unique $M_i\in\operatorname{GL}_m$ s.t.
\[
F_i=[L_i, R_i]M_i,
\]
and
\[
    \boxed{
    \Omega_P(\mathcal{T}_i)
    =
   \det(M_i)\det(PR_i)\Alt(L_i)\in\bigwedge^2\mathcal{L}_i.
    }
\]
Furthermore, if $\det(PR_i)\neq0$, then
\[
    \im\Omega_P(\mathcal{T}_i)=\mathcal{L}_i.
\]
If $\det(PR_i)\neq0$ for every $i$, then
\[
    \im\Omega_P(\mathcal{T})
    =
    \bigoplus_{i=1}^r\mathcal{L}_i
    =
    \mathcal{L}.
\]
\end{corollary}

\begin{proof}
By Lemma \ref{lem: alternating properties}, 
\[
    \Omega_P(\mathcal{T}_i)=\iota_P\Alt(F_i)=\det(M_i)\iota_P\Alt([L_i,R_i]).
\]
By Lemma \ref{lem: generic-section}, $\dim\mathcal{L}_i=2$ for all $i$. Since every $P\in\Gr(d,\mathcal{L}^\perp)$ annihilates $\mathcal{L}$, all terms in Equation \eqref{eq: contraction formula} vanish except the one obtained by deleting the two columns of $L_i$.
\[
\implies 
 \Omega_P(\mathcal{T}_i)=
   \det(M_i)\det(PR_i)\Alt(L_i).
\]
If $\det(PR_i)\neq 0$, then $\im\Omega_P(\mathcal{T}_i)=\im\Alt(L_i)=\mathcal{L}_i$. Now suppose this holds for all $i$. By linearity of the contraction and Lemma \ref{lem: generic-section}, 
\[
\im\Omega_P(\mathcal{T})\subseteq\bigoplus_{i=1}^r\mathcal{L}_i=\mathcal{L}.
\]
Together with skew-symmetry of the contraction we get,
\[
\im\Omega_P(\mathcal{T})^\perp=\ker\Omega_P(\mathcal{T})\subseteq \bigcap_{i=1}^r\mathcal{L}_i^\perp=\mathcal{L}^\perp,
\]
where containment is due to directness. Together, the two inclusions imply that $\im\Omega_P(\mathcal{T})=\mathcal{L}$.      
\end{proof}

\begin{lemma}[Spectral separation]
\label{lem: spectral-separation}
For generic $P\in\Gr(d,X^*)$ set $\mathcal V=\ker P$ and let $\mathcal{L}_i, \mathcal{R}_i$, and $\mathcal{L}$ be as in Lemma \ref{lem: generic-section}. Fix a basis $R_i$ for $\mathcal{R}_i$. For generic  $Q\in\Gr(d,\mathcal{L}^{\perp})$, we have
\[
    \det(PR_i)\neq0,
    \qquad
    \det(QR_i)\neq0
    \qquad\text{for every }i,
\]
and the scalars
\[
    \rho_i
    :=
    \frac{\det(QR_i)}{\det(PR_i)}
\]
are pairwise distinct. 
\end{lemma}

\begin{proof}
Since $\mathcal{L}_i:= \ker P\cap\mathcal{F}_i$ the restriction of $P$ to $\mathcal{R}_i$ is injective. Hence for any basis,
\[
    \det(PR_i)\neq0.
\]
It follows that $Q\mapsto\det(QR_i)$ is not identically zero on
$\Gr(d,\mathcal{L}^{\perp})$, and therefore $\det(QR_i)\neq0$ for generic $Q$.

Fix $i\neq j$ and consider
\[
    f_{ij}(Q)
    :=
    \det(QR_i)\det(PR_j)
    -
    \det(QR_j)\det(PR_i).
\]
By Lemma \ref{lem: generic-section}, there exists
$v_{ij}\in\mathcal{R}_j$ such that $\mathcal{R}_i
    \oplus\langle v_{ij}\rangle\oplus\mathcal{L}$ is direct. Hence for generic $Q_0\in\Gr\bigl(d,   \mathcal{L}^{\perp}\cap v_{ij}^{\perp}
    \bigr)$, $\det(Q_0R_i)\neq0$, while $\det(Q_0R_j)=0$. Thus
\[
    f_{ij}(Q_0)
    =
    \det(Q_0R_i)\det(PR_j)
    \neq0,
\]
and $f_{ij}$ is not the zero polynomial.
Since there are finitely many pairs $i\neq j$, a generic $Q$ satisfies
\[
    f_{ij}(Q)\neq0
    \qquad
    \text{for all }i\neq j,
\]
or equivalently, $\rho_i\neq\rho_j$.
\end{proof}

\section{The LP condition is polynomial-time verifiable}
\label{sec: verification}

We prove that Condition \eqref{con: m-LP} is verifiable in $\mathrm{poly}(D, r)$ time. This fact is an immediate consequence of polynomial-time submodular-function minimization.

\begin{claim}
\label{claim: LP verification}
Given $A^{(1)},\dots,A^{(m)}$, with $A^{(j)}\in\mathbb{R}^{n_j\times r}$, Condition~\eqref{con: m-LP} can be verified deterministically
 in time $\mathrm{poly}(D,r)$ where $D:=\sum_{j=1}^m n_j$.
\end{claim}

\begin{proof}

Define
\[
f(S)
:=
\sum_{\ell=1}^m \rank A^{(\ell)}[S]
-2|S|.
\]
Each matrix-rank function is a matroid rank function and hence submodular, while $2|S|$ is modular. Therefore, $f$ is submodular.

Condition~\eqref{con: m-LP} is equivalent to

$$
f(S)\geq m-1
\qquad
\text{for every }S\subseteq[r]\text{ with }|S|\geq 2.
$$

For each pair $i\neq j$, define on the ground set
$[r]\setminus\{i,j\}$,
\[
f_{i,j}(R):=f(R\cup\{i,j\}).
\]
The function $f_{i,j}$ is again submodular, and its ground set has size $r-2$.  Submodular-function minimization can be performed in time polynomial in the ground-set size and in the time required to evaluate the function (see \cite{Schrijver2000Submodular}). Evaluating $f_{i,j}(R)$ requires computing the ranks of the matrices $A^{(\ell)}[R\cup\{i,j\}]\in\mathbb R^{n_\ell\times (|R|+2)}$. Using Gaussian elimination, this takes \[ O(Dr^2) \] arithmetic operations. Hence, for every pair $i\neq j$, we can compute
\[
\min_{S:S\supseteq\{i,j\}} f(S)
\]
in time polynomial in $r$ and $D$. 

Condition~\eqref{con: m-LP} holds if and only if this minimum is at least $m-1$ for every pair $i\neq j$. Since there are $\binom{r}{2}$ such pairs, Condition~\eqref{con: m-LP} can be verified deterministically in $\operatorname{poly}(r,D)$ time.
\end{proof}

\section{Uniqueness proof}
\label{sec: uniqueness}

\begin{lemma}
\label{lem: uniqueness}
Let
\[
\mathcal{T}
=
\llbracket A^{(1)},\dots,A^{(m)}\rrbracket
\in
\mathbb{R}^{n_1\times\dots\times n_m}
\]
be an $r$-decomposition satisfying Condition~\eqref{con: m-LP}. Then
$\llbracket A^{(1)},\dots,A^{(m)}\rrbracket$ is the unique rank
decomposition of $\mathcal T$.
\end{lemma}

\begin{proof}
We first show that $\rank(\mathcal T)=r$. Let
\[
\mathcal T
=
\sum_{\ell=1}^s
\widetilde a_\ell^{(1)}
\otimes\cdots\otimes
\widetilde a_\ell^{(m)}
\]
be any decomposition of $\mathcal T$ into $s$ rank-one tensors. By
linearity,
\[
\Omega_P
=
\sum_{\ell=1}^s
\Omega_P(\widetilde{\mathcal{T}}_\ell).
\]

For generic \(P\), \(\dim(\ker P\cap\widetilde{\mathcal F}_i)=2\), and the same calculation as in Corollary~3.6 shows that
\[
\im\Omega_P(\widetilde{\mathcal T}_i)\subseteq \ker P\cap\widetilde{\mathcal F}_i.
\]
\[
\implies 
2r
=
\rank(\Omega_P)
\leq
\sum_{\ell=1}^s
\rank\Omega_P(\widetilde{\mathcal{T}}_\ell)
\leq 2s.
\]
Thus $s\geq r$. Since an $r$-decomposition is given,
\[
\rank(\mathcal T)=r.
\]

It remains to prove uniqueness. This does not follow immediately from the recovery guarantee, since a hypothetical second rank decomposition is not known a priori to satisfy Condition~\eqref{con: m-LP}. 
Suppose
\[
\mathcal T
=
\sum_{i=1}^r
\widetilde a_i^{(1)}
\otimes\cdots\otimes
\widetilde a_i^{(m)}
\]
is an alternative rank decomposition. Then
\[
\Omega_P=\sum_{i=1}^r\Omega_P(\widetilde{\mathcal{T}}_i),
\qquad
\Omega_Q=\sum_{i=1}^r\Omega_Q(\widetilde{\mathcal{T}}_i).
\]
Since both $\Omega_P$ and $\Omega_Q$ have rank $2r$ (this follows from the \eqref{con: m-LP}-compliant decomposition), while every summand has rank at most $2$, necessarily
\begin{equation}
\label{eq: uniqueness eq 1}
\rank\Omega_P(\widetilde{\mathcal{T}}_i)
=
\rank\Omega_Q(\widetilde{\mathcal{T}}_i)
=
2
\qquad\text{for every }i,
\end{equation}
and
\[
\mathcal{L}=\bigoplus_{i=1}^r\widetilde{\mathcal{L}}_i,
\qquad
\widetilde{\mathcal{L}}_i:=\im\Omega_P(\widetilde{\mathcal{T}}_i).
\]
Setting bases $\widetilde{R}_i$ for complements $\widetilde{\mathcal{R}}_i$ (i.e. $\widetilde{\mathcal{F}}_i=\widetilde{\mathcal{L}}_i\oplus \widetilde{\mathcal{R}}_i$), the same calculation used in Corollary~\ref{cor: L-adapted} implies that
\[
\Omega_Q(\mathcal{T})=\sum_{i=1}^r\widetilde{\rho}_i\Omega_P(\widetilde{\mathcal{T}}_i)
\]
for nonzero \(\widetilde{\rho}_i\). Consequently, as in the proof of Theorem~\ref{thm: algorithmic LP}, \(L^\top(\widetilde{\mathcal L}_i)\) is the \(\widetilde{\rho}_i\)-eigenspace of \(\Phi\). Therefore, \(\widetilde{\rho}_i=\rho_{\sigma(i)}\) for some permutation \(\sigma\), and
$\widetilde{\mathcal L}_i=\mathcal L_{\sigma(i)}$.

Uniqueness now follows, since the same spectral projector isolates $\Omega_P(\widetilde{\mathcal{T}}_i)$.
Factoring the nonzero rank-one blocks of these matrices recovers the same
column vectors in every mode, up to scaling. Thus the two decompositions agree up to permutation and scaling of the columns.
\end{proof}

\section{Digest: a geometric perspective}
\label{sec: geometric perspective}

In this section we compare the conceptual framework of the Jennrich algorithm \cite{leurgans1993decomposition} with this work. The algorithm can be understood as a relativized ``Alternating Jennrich'', where instead of generically contracting the input tensor, contractions are chosen from a restricted space, and are applied to the alternating form of the tensor, rather than directly to the input tensor as in the Jennrich case. We start by reviewing Jennrich's algorithm. 

\subsection{Jennrich's algorithm}

Given a an $r$-decomposition 
\[
\mathcal{T}=\sum_{i=1}^r a_i\otimes b_i\otimes c_i=\llbracket A,B, C\rrbracket\in\mathcal{A}\otimes\mathcal{B}\otimes\mathcal{C},
\]
we say that $\mathcal{T}$ satisfies the Jennrich condition if $A$ and $B$ have full column rank, and $C$ has no collinear columns. After the usual mode compression and a generic projection of the \(C\)-mode, we may assume \(\dim\mathcal A=\dim\mathcal B=r\) and \(\dim\mathcal C=2\).\footnote{Any tensor can be efficiently compressed, by restricting each mode to the column space of the corresponding mode unfolding, and a decomposition of the compressed tensor can be efficiently lifted to a decomposition of the original tensor; see, e.g., \cite{KoldaBader2009,DeLathauwer2000}.}
Denote $\mathcal{T}_i=a_i\otimes b_i\otimes c_i\in\mathcal{A}\otimes\mathcal{B}\otimes\mathcal{C}$.
For any $x\in \mathcal{C}^*$ and $i\in[r]$,
\[
\mathcal{T}_i\times_\mathcal{C} x=x(c_i) \left(a_i\otimes b_i\right)
\]
so the image of the contraction over $\mathcal{C}^*$ is the 1-dimensional space
\[
\mathcal E_i:=\langle a_i\otimes b_i\rangle.
\]
The Jennrich Condition guarantees that the contracted tensor $T_x:=\mathcal{T}\times_\mathcal{C} x= \sum_{i=1}^r x(c_i) \left(a_i\otimes b_i\right)$ is generically invertible and lives in
\[
\mathcal E=\bigoplus_{i=1}^r \mathcal E_i.
\]

Access to a non-zero element of $\mathcal E_i$ allows for immediate recovery of $a_i$ and $b_i$ up to scaling. To isolate this space, Jennrich's algorithm computes
\[
\Phi:=T_y T_x^{-1}=A\operatorname{diag}(\frac{y C}{x C})A^{-1}.
\]
Using $\Phi$ and its eigenvalues $\lambda_i:=\frac{y(c_i)}{x(c_i)}$, the algorithm proceeds by computing the $i$-th eigenprojector $\Pi_i$. Applying the eigenprojector isolates
\[
\Pi_i T_x=x(c_i) a_ib_i^\top\in \mathcal E_i,
\]
from which $a_i$ and $b_i$ are recovered up to scaling.
For this eigendecomposition to work, the eigenvalues $\lambda_i$ must be distinct. This is guaranteed (for generic $x,y\in \mathcal{C}^*$), due to the no-collinearity  requirement.

\subsection{A comparison to our algorithm}

Recall that
\[
X:=\mathcal A^{(1)}\oplus\cdots\oplus\mathcal A^{(m)},
\qquad
D:=\dim X=\sum_{j=1}^m n_j,\qquad d=m-2.
\]

As in Jennrich's algorithm, our algorithm begins by generating contractions. Under Condition~\eqref{con: m-LP}, however, the contractions are applied to the alternating representation of the tensor (Definition~\ref{def: alternating contraction}), rather than directly to the input tensor. Specifically, we choose
\(\mathcal P\in\Gr(d,X^*)\) generically and compute\footnote{We use the notation \(\Omega_{\mathcal P}\) in this discussion because the construction is basis-independent: if \(P\) represents \(\mathcal P\), then \(\Omega_{AP}\propto\Omega_P\) for every invertible \(A\).}
\[
\Omega_{\mathcal P}
=
\iota_{\mathcal P}\bigl(\Alt(\bar{\mathcal T})\bigr)
\in\mathbb R^{D\times D}.
\]

We then choose
\[
\mathcal Q\in\Gr(d,\ker\Omega_{\mathcal P})
\]
generically and compute \(\Omega_{\mathcal Q}\).

The apparent asymmetry between \(\mathcal P\) and \(\mathcal Q\) is imposed by the algorithmic setting rather than by the underlying geometry. Conceptually, one should instead first choose
\[
\mathcal V\in\Gr(D-d,X)
\]
generically and define\footnote{For definition of $\mathcal{F}_i$ refer to the \hyperref[sec: preliminaries]{preliminaries}.}
\[
\mathcal L_i:=\mathcal V\cap\mathcal F_i,
\qquad
\mathcal L:=\sum_{i=1}^r\mathcal L_i.
\]
One would then choose both \(\mathcal P\) and \(\mathcal Q\) generically in
\[
\Gr(d,\mathcal L^\perp).
\]
This scheme cannot be implemented directly, since the spaces \(\mathcal F_i\), and hence \(\mathcal L\), are unknown. The algorithm instead chooses \(\mathcal P\) first and sets
\[
\mathcal V:=\ker\mathcal P.
\]
The same \(\mathcal P\) then belongs to \(\Gr(d,\mathcal L^\perp)\), although it is no longer generic there, while \(\mathcal Q\) can be sampled generically from this Grassmannian. The two viewpoints lead to the same spectral structure.

Before comparing this construction with Jennrich's algorithm, it is useful to understand geometrically what an alternating contraction does to a single rank-one term. While \(C\)-mode contractions of a rank-one tensor
\[
\mathcal T_i\in\mathcal A\otimes\mathcal B\otimes\mathcal C
\]
take values in the fixed one-dimensional space \(\mathcal E_i\), the alternating contraction is naturally described through its action on the associated subspace \(\mathcal F_i\subset X\). The alternating tensor
\[
\Alt(F_i)\in\bigwedge^mX
\]
is, by definition, a Pl\"ucker representative of \(\mathcal F_i\). For \(P\in\Gr(d,X^*)\),
\[
\Omega_P(\mathcal T_i)
=
\iota_P\Alt(F_i)
\in
\bigwedge\nolimits^2(\ker P\cap\mathcal F_i),
\]
and, whenever this contraction is nonzero, it is a Pl\"ucker representative of the two-dimensional subspace
\[
\ker P\cap\mathcal F_i.
\]

Schematically,
\[
\begin{tikzcd}[column sep=large,row sep=large]
\Alt(\bar{\mathcal T}_i)=\Alt(F_i)
    \arrow[r, "\iota_P"]
&
\Omega_P(\mathcal T_i)=\iota_P\Alt(F_i)
\\
\mathcal F_i
    \arrow[u, "\mathrm{Pl\ddot{u}cker}"]
    \arrow[r, "\mathcal F\mapsto \ker P\cap\mathcal F"']
&
\ker P\cap\mathcal F_i
    \arrow[u, "\mathrm{Pl\ddot{u}cker}"']
\end{tikzcd}
\]

We can now make the comparison with Jennrich more precise. Recall that
\[
\mathcal E_i
=
\operatorname{span}
\left\{
\mathcal T_i\times_{\mathcal C}x
\;\middle|\;
x\in\mathcal C^*
\right\}
=
\langle a_i\otimes b_i\rangle
\]
is one-dimensional. By contrast, if alternating contractions are allowed over the entire Grassmannian, then
\[
\operatorname{span}
\left\{
\Omega_P(\mathcal T_i)
\;\middle|\;
P\in\Gr(d,X^*)
\right\}
=
\bigwedge\nolimits^2\mathcal F_i,
\]
a space of dimension
\[
\binom{m}{2}.
\]
The restriction to \(\Gr(d,\mathcal L^\perp)\) changes this. Relative to the fixed generic space \(\mathcal V\), for every
\(P\in\Gr(d,\mathcal L^\perp)\),
\[
\Omega_P(\mathcal T_i)
\in
\bigwedge\nolimits^2\mathcal L_i.
\]
Since \(\dim\mathcal L_i=2\),
\[
\dim\bigwedge\nolimits^2\mathcal L_i=1,
\]
and for generic \(P\in\Gr(d,\mathcal L^\perp)\),
\[
\bigwedge\nolimits^2\mathcal L_i
=
\langle\Omega_P(\mathcal T_i)\rangle.
\]
Thus restricting the contraction space to \(\Gr(d,\mathcal L^\perp)\) recovers, componentwise, the same one-dimensional structure that is present automatically for ordinary contractions of a rank-one 3-way tensor. Importantly, this observation is independent of Condition~\eqref{con: m-LP}; the role of Condition~\eqref{con: m-LP} is to guarantee the direct-sum geometry and spectral separation needed for recovery.

Indeed, in Jennrich's setting the relevant condition guarantees that the spaces
\(\mathcal E_1,\ldots,\mathcal E_r\) are independent. In the present setting, Condition~\eqref{con: m-LP} implies that, for generic
\(\mathcal V\in\Gr(D-d,X)\),
\[
\mathcal L
=
\bigoplus_{i=1}^r\mathcal L_i.
\]
Moreover, for generic
\(P\in\Gr(d,\mathcal L^\perp)\),
\[
\im\Omega_P(\mathcal T_i)=\mathcal L_i,
\]
and hence
\[
\im\Omega_P(\mathcal T)
=
\bigoplus_{i=1}^r
\im\Omega_P(\mathcal T_i).
\]
Thus the direct-sum geometry needed by the spectral step has an analogue in the two algorithms, although it arises for different reasons.

The second parallel is spectral separation. In Jennrich's algorithm, for generic \(x,y\in\mathcal C^*\),
\[
\mathcal T\times_{\mathcal C}y
=
\sum_{i=1}^r
\rho_i
\bigl(\mathcal T_i\times_{\mathcal C}x\bigr),
\]
with the scalars \(\rho_i\) pairwise distinct. In our setting, for generic
\(P,Q\in\Gr(d,\mathcal L^\perp)\),
\[
\Omega_Q(\mathcal T)
=
\sum_{i=1}^r
\rho_i\Omega_P(\mathcal T_i),
\]
again with pairwise distinct nonzero scalars \(\rho_i\). In both cases, these distinct coefficients break the symmetry between the tensor terms and permit their recovery through an eigendecomposition.

What is substantially different is the geometric mechanism that produces this structure. To understand the role of Condition~\eqref{con: m-LP}, consider the linear polymatroid represented by the family
\(\{\mathcal F_i\}_{i\in[r]}\), with rank function
\[
R(S)
:=
\dim\left(\sum_{i\in S}\mathcal F_i\right),
\]
and define the shifted set function
\[
\widetilde R(S):=R(S)-d.
\]
A theorem of Lov\'asz (Theorem~\ref{thm: lovasz-dilworth}) states that, for generic
\(\mathcal V\in\Gr(D-d,X)\),
\[
\dim\left(
\sum_{i\in S}
(\mathcal V\cap\mathcal F_i)
\right)
=
\min_{\Pi\vdash S}
\sum_{B\in\Pi}\widetilde R(B),
\qquad
S\subseteq[r].
\]
Thus the dimensions of the intersected spaces
\(\mathcal V\cap\mathcal F_i\) are controlled by the dimensions of the original family \(\{\mathcal F_i\}\); equivalently, these intersections realize the corresponding Dilworth truncation (see~\cite{antolini2026dilworth} for this terminology). Condition~\eqref{con: m-LP}, in the form~\eqref{eq: m-LP F-form}, provides precisely the lower bounds on \(R\) needed to deduce
\[
\mathcal L
=
\bigoplus_{i=1}^r\mathcal L_i.
\]
A second application of the same theorem gives the stronger geometric separation needed to prove that the spectral coefficients \(\rho_i\) are pairwise distinct; see Lemma~\ref{lem: generic-section}.

Consequently, Condition~\eqref{con: m-LP} supplies the two structural properties required by the spectral step:
\begin{enumerate}
    \item
    \[
    \im\Omega_P(\mathcal T)
    =
    \bigoplus_{i=1}^r
    \im\Omega_P(\mathcal T_i)
    \]
    for generic \(P\in\Gr(d,\mathcal L^\perp)\);
    \item for generic \(P,Q\in\Gr(d,\mathcal L^\perp)\), the coefficients
    \(\rho_1,\ldots,\rho_r\) are nonzero and pairwise distinct.
\end{enumerate}

\paragraph{Comparison with previous Kruskal algorithms.}
The algorithms of~\cite{domanov2014canonical} and~\cite{bhargava2026algorithmic} generalize Jennrich in a different direction. In the Jennrich setting, the rank of an ordinary \(C\)-mode contraction drops when the contraction vector lies in one of the hyperplanes $c_i^\perp$, perpendicular to the $i$-th column of $C$.

Under the more general Kruskal condition, the relevant contraction space may have dimension \(\Omega(r)\), and the intersections of the hyperplane arrangement that are useful for recovery may involve \(\Omega(r)\) hyperplanes. Once these intersections are identified, the decomposition can be reduced to Jennrich by projecting to suitable one-dimensional intersections of the arrangement. The difficulty is that finding such highly nongeneric points in a high-dimensional contraction space is costly, leading to super-polynomial running time in general.

The alternating construction avoids this search over the hyperplane arrangement. After passing to the alternating representation and restricting to \(\Gr(d,\mathcal L^\perp)\), the required spectral separation is visible along a generic pencil
\[
\Omega_P+t\Omega_Q.
\]
Thus, in the previous approaches the useful structure is revealed by searching for special rank-dropping contractions, whereas in our approach it appears along a generic one-parameter family of alternating contractions. This distinction is ultimately responsible for the different complexity behavior of the two approaches.

\section{Appendix}
\label{sec: appendix}

\begin{theorem}[Lov\'asz {\cite{lovasz1977flats}}; see also {\cite[Theorem~2.10]{antolini2026dilworth}}]
\label{thm: lovasz-dilworth}
Let $\{\mathcal F_i\subseteq X\}_{i\in[r]}$ be a collection of linear
subspaces, and let
\[
    R(S):=\dim\left(\sum_{i\in S}\mathcal F_i\right),
    \qquad S\subseteq[r],
\]
be the rank function of the represented linear polymatroid. Fix $d\geq 0$ such that $\dim\mathcal F_i\geq d$ for every $i\in[r]$, and
define
\[
    \widetilde R(S):=R(S)-d
\]
for every nonempty $S\subseteq[r]$.

For generic $\mathcal V\in\Gr(\dim X-d,X)$, let
\[
    \mathcal L_i:=\mathcal V\cap\mathcal F_i,
    \qquad i\in[r].
\]
Then, for every nonempty $S\subseteq[r]$,
\begin{equation}
\label{eq: lovasz}
\tag{L}
    \dim\left(\sum_{i\in S}\mathcal L_i\right)
    =
    \min_{\Pi\vdash S}
    \sum_{B\in\Pi}\widetilde R(B),
\end{equation}
where the minimum is over all partitions $\Pi$ of $S$ into nonempty
parts. Equivalently, $\{\mathcal L_i\}_{i\in[r]}$ is a linear
representation of the Dilworth truncation $\widetilde R^{\,D}$.
\end{theorem}

See \cite{antolini2026dilworth} for a modern formulation of the theorem, the terminology used here, and an accessible proof.

\begin{proof}[Proof of Lemma \ref{lem: alternating properties}]
\label{proof: alternating properties}
For Part~(1), write
\[
N=[n_1,\ldots,n_m],
\qquad
EN=[f_1,\ldots,f_m],
\]
where
\[
f_j=En_j=\sum_{s=1}^m N_{sj}e_s.
\]
By multilinearity,
\[
\begin{aligned}
\Alt(EN)
&=
f_1\wedge\cdots\wedge f_m\\
&=
\sum_{s_1,\ldots,s_m\in[m]}
\left(\prod_{j=1}^m N_{s_jj}\right)
e_{s_1}\wedge\cdots\wedge e_{s_m}.
\end{aligned}
\]
Terms with repeated indices vanish by alternation, so only permutations of
$[m]$ remain. Hence
\[
\begin{aligned}
\Alt(EN)
&=
\sum_{\sigma\in S_m}
\left(\prod_{j=1}^m N_{\sigma(j)j}\right)
e_{\sigma(1)}\wedge\cdots\wedge e_{\sigma(m)}\\
&=
\left(
\sum_{\sigma\in S_m}
\sgn(\sigma)
\prod_{j=1}^m N_{\sigma(j)j}
\right)
\Alt(E)\\
&=
\det(N)\Alt(E).
\end{aligned}
\]
If $N\in\operatorname{GL}_m$, then $\det(N)\neq0$, and therefore
$\langle\Alt(EN)\rangle=\langle\Alt(E)\rangle$.

For Part (2) we first prove the following special case. Let
\[
F=[f_1,\dots,f_d]\in X^{\times d},
\qquad
P=
\begin{bmatrix}
p_1\\
\widehat P
\end{bmatrix},
\]
where $\widehat P$ has rows $p_2,\dots,p_d$. Then
\begin{equation}
\label{eq: full contraction}
\iota_P\Alt(F)=\det(PF).
\end{equation}
Indeed, by induction on $d$,
\[
\begin{aligned}
\iota_P\Alt(F)
&=
\iota_{\widehat P}
\left(
\sum_{i=1}^d
(-1)^{i-1}p_1(f_i)\Alt(F[\widehat i])
\right)\\
&=
\sum_{i=1}^d
(-1)^{i-1}p_1(f_i)
\det\!\left(\widehat P F[\widehat i]\right)\\
&=
\det(PF),
\end{aligned}
\]
where the last equality is the Laplace expansion of $\det(PF)$ along
its first row.

Now let $E=[e_1,\dots,e_m]$ and $d=m-2$. For any $1\le j<k\le m$, let
\[
E(j,k):=
\bigl[E[\widehat{j,k}]\,\big|\,e_j\,\big|\,e_k\bigr],
\]
where the columns of $E[\widehat{j,k}]$ retain their original order.
The permutation taking $E$ to $E(j,k)$ has sign $(-1)^{j+k+1}$, therefore by Part~(1),
\[
\Alt(E)
=
(-1)^{j+k+1}\Alt(E(j,k)).
\]

Gathering the terms in which $e_j,e_k$ remain uncontracted, their contribution to
$\iota_P\Alt(E)$ is
\[
(-1)^{j+k+1}
\iota_P\Alt(E[\widehat{j,k}])\Alt(E[j,k])
\overset{\eqref{eq: full contraction}}{=}
(-1)^{j+k+1}
\det(PE[\widehat{j,k}])\Alt(E[j,k]).
\]

 Summing over all $1\le j<k\le m$ gives
\[
\iota_P\Alt(E)=
\sum_{1\leq j<k\leq m}(-1)^{j+k+1}\det\!\left(PE[\widehat{j,k}]\right)
\Alt(E[j,k]),
\]
as claimed.

\end{proof}

\subsection{Relation to Kruskal's condition}

For a matrix $A$, its \emph{Kruskal rank}, denoted $k(A)$, is the
largest integer $k$ such that every set of $k$ columns of $A$ is linearly
independent. For an $r$-decomposition
\[
    \mathcal T=\llbracket A^{(1)},\ldots,A^{(m)}\rrbracket
\]
let
\[
    k_j:=k(A^{(j)}), \qquad j\in[m].
\]
Kruskal's condition is\footnote{\cite{kruskal1977three} proved identifiability for $3$-way tensors. The $m$-way generalization is due to \cite{sidiropoulos2000uniqueness}.}
\begin{equation}
\label{eq: kruskal condition}
    \sum_{j=1}^m k_j \ge 2r+m-1.
\end{equation}

\begin{proposition}
\label{prop: kruskal implies LP}
Kruskal's condition implies Condition~\eqref{con: m-LP}.
\end{proposition}

\begin{proof}
Let
\[
    k_j:=k(A^{(j)}), \qquad j\in[m],
\]
and suppose Kruskal's condition holds, i.e.,
\[
    \sum_{j=1}^m (k_j-1)\ge 2r-1.
\]
Fix $S\subseteq[r]$ with $s:=|S|\ge2$. By the definition of Kruskal
rank,
\[
    \rank A^{(j)}[S]\ge \min\{k_j,s\}.
\]
Let $T=\{j\in[m]\;|\; k_j>s\}$. If $t:=|T|>2$, then

\[
    \left(\sum_{j=1}^m \min\{k_j,s\}\right)-m=\sum_{j=1}^m \min\{k_j-1,s-1\}\ge 2(s-1)+1,
\]
which implies Condition \eqref{con: m-LP}. Otherwise, $t\leq 2$, so
\[
\sum_{j=1}^m\min\{k_j,s\}=\sum_{j=1}^m k_j+\sum_{j\in T}(s-k_j)\geq
2r+m-1+2(s-r).
\]
Either way, \eqref{con: m-LP} is satisfied.
\end{proof}

\paragraph{Acknowledgments.} I am grateful to Leonard Schulman for helpful discussions and feedback on this work. The author was supported in part by NSF CCF-2321079.

\paragraph{AI methodology.}
The author used ChatGPT (OpenAI) throughout the preparation of this manuscript as an aid for mathematical brainstorming, checking and refining arguments, improving exposition, LaTeX formatting, and locating relevant literature. All mathematical statements, proofs, citations, and conclusions appearing in the final manuscript were independently checked and are the responsibility of the author.

\bibliographystyle{alphaurl}
\bibliography{references}

\end{document}